\documentclass{llncs}

\usepackage[dvips,bookmarksopen]{hyperref}
\usepackage{comment}
\usepackage{tikz}
\usetikzlibrary{positioning}
\usepackage{amsmath}
\usepackage{amssymb}
\usepackage[T1]{fontenc}
\usepackage{todonotes}

\newcommand{\Prob}{\mathbb{P}}
\newcommand{\APP}{\mathsf{APP}}
\newcommand{\TPP}{\mathsf{TPP}}
\newcommand{\spa}{\mathsf{space}}
\newcommand{\SWAPP}{\mathsf{SWAPP}}

\begin{document}

\title{Streaming algorithms for products of probabilities}

\author{Markus Lohrey \and Leon Rische \and Louisa Seelbach \and Julio Xochitemol}

\institute{Universit\"at Siegen, Department ETI, 57076 Siegen, Germany}



\maketitle

\begin{abstract}
We consider streaming algorithms for approximating a product of input probabilities up to multiplicative error of $1-\epsilon$. 
It is shown that every randomized streaming algorithm for this problem needs space $\Omega(\log n + \log b - \log \epsilon - \Theta(1))$,
where $n$ is length of the input stream and $b$ is the bit length of the input numbers. This matches an upper bound from Alur et al.~up to 
a constant multiplicative factor. Moreover, we consider the threshold problem, where it is asked whether the product of the
input probabilities is below a given threshold. It is shown that every  randomized streaming algorithm for this problem needs space 
$\Omega(n \cdot b)$. Finally, also the sliding window variant of the  approximation problem for a product of input probabilities is considered.
\end{abstract}

 \section{Introduction}
 
 The streaming model is one of the central concepts in data science to deal with huge data sets, where random access is no longer feasible.
 It is based on the processing of an input stream of data values. Every data value has to be processed immediately when it arrives. For a good
 introduction to streaming algorithms see \cite{BHK2020,Muthukrishnan05}. In this paper we consider the problem of computing probabilities in the streaming model.
 More precisely, consider a stream of probabilities $q_1, q_2, \ldots, q_n$ where every $q_i$ is given as a fraction $r_i/s_i$ of integers $r_i, s_i$
 with $r_i \geq 0$ and $1 \leq r_i \leq s_i$. The goal is to compute a good approximation of the product probability $\prod_{1 \leq i \leq n} q_i$.
 This problem has been first studied by Alur et al.~in an automata theoretic setting \cite{AlurCJK20}. It was shown that
 a $(1-\epsilon)$-approximation\footnote{A $(1-\epsilon)$-approximation of a number $x$ is any number $x'$ with
 $(1-\epsilon)x <x' < x/(1-\epsilon)$.}
 of the product $\prod_{1 \leq i \leq n} q_i$ can be computed by a deterministic streaming algorithm in space  $\log_2(n^2 b / \epsilon) = 2 \log n + \log b - \log \epsilon$. Here, $n$ is the length of the stream and $b$ is the maximal bit length of the numerators $r_i$ and denominators $s_i$ in the stream of probabilities $q_i = r_i/s_i$.\footnote{\label{fn1}The space bound stated in \cite{AlurCJK20} is actually larger than $2 \log n + \log b - \log \epsilon$ 
since the authors in   \cite{AlurCJK20} also count the space for internal computations; see Section~\ref{sec streaming} for more details on this.}

 Our first result is a lower bound that matches the upper bound $2 \log n + \log b - \log \epsilon$ up to a constant factor.
 This is also true for randomized streaming algorithms that produce a $(1-\epsilon)$-approximation of the product
$\prod_{1 \leq i \leq n} q_i$ with high probability. Our lower bound only holds if $\log (1/\epsilon) \le b  <  \log(1/(1-\epsilon)) \cdot n/2$.
This is a reasonable setting. If we consider the approximation parameter $\epsilon$ as a small constant then 
the requirement becomes $\Omega(1) \leq b \leq \mathcal{O}(n)$. Since the length $n$ of the stream is usually very large,
it is reasonable to assume $b \leq \mathcal{O}(n)$. In a realistic scenario one would even require 
$b \leq \mathcal{O}(\log n)$.

Our second result is a new lower bound for the threshold problem: the input consists again of a stream of probabilities 
 $q_1, q_2, \ldots, q_n$, where $q_i = r_i/s_i$ and the integers $r_i$ and $s_i$ need at most $b$ bits.
 In addition, there is a threshold probability $t = r/s$, where $r$ and $s$ also need at most $b$ bits. The question is whether
 $\prod_{1 \leq i \leq n} q_i < t$ holds.
 It was shown in \cite[Theorem~A.1]{AlurCJK20} that every deterministic streaming algorithm for the general threshold problem, where 
 the $q_i$ and $t$ are not restricted to the interval $[0,1]$, has to store $\Omega(n)$ many bits when  $b = \log n$.
 It was asked whether this lower bound also holds when all $q_i$ and $t$ are probabilities. We prove that this is indeed the case
 and moreover improve the lower bound from $\Omega(n)$ to $\Omega(n \cdot b)$ if $b \geq \Omega(\log n)$ for randomized streaming 
 algorithms. Note that the threshold problem can be always solved in space $\Omega(n \cdot b)$ by naively computing the product   $\prod_{1 \leq i \leq n} q_i$. On the other hand, if $n$ is very large compared to $b$, then there is a better (deterministic) algorithm that 
 only needs $\mathcal{O}(2^b)$ bits; see Theorem~\ref{thm-threshold-upper}.
  
  In the final section we consider the product approximation problem for rational probabilities in the sliding window setting. Here, the goal
  is to maintain a $(1-\epsilon)$-approximation of the product of the last $m$ probabilities from the input stream. The parameter $m$ is called the window size. 
  We provide a space lower bound of $\Omega(m \cdot (\log b - \log \epsilon - \Theta(1)))$ for randomized sliding window algorithms, where as before $b$ is the bit length of the input probabilities.
 Currently, we are not able to match this lower bound by an upper bound. In Section~\ref{sec-SW} we provide two incomparable upper bounds:
 $2 m b$ (which is trivial) and $m \cdot (\log m + \log b - \log \epsilon)$. 
 
 \subsection{Related work}

 The first paper that studied the sliding window model formally was \cite{DatarM02}. It is shown there that for a stream of 
 positive integers of bit length $b$,
 a $(1-\epsilon)$-approximation of the sum of the numbers in a sliding window of length $m$ can be maintained in space
 $\mathcal{O}( (1/\epsilon) (\log m + b) \log m)$. A matching lower bound is shown as well. A survey on sliding window algorithm can be found in \cite{Aggarwal07}.

  \section{Preliminaries}

For a possibly infinite set $A$ we denote with $A^*$ the set of all finite words (or sequences) over the set $A$. We 
write $A^{\le n}$ for the subset of words of length at most $n$.
With $\log$ we always denote the logarithm to the base two.
For a positive rational number $q = r/s$ with $r,s \in \mathbb{N}$ and $\mathsf{gcd}(r,s)=1$ we define the {\em bit size} of $q$ as
$\Vert q \Vert = \max\{\lceil \log r \rceil, \lceil \log s \rceil\}$. In addition, we define $\Vert 0 \Vert = 0$.
Let $\Prob = \{ q \in \mathbb{Q} : 0 \le q \le 1 \}$ be the set of rational probabilities and 
$\Prob_b = \{ q \in \Prob : \Vert q \Vert \leq b \}$ be the set of rational probabilities of bit size at most $b$.

\subsection{Streaming algorithms} \label{sec streaming}

We only consider streaming problems where the input values are 
rational probabilities and there is a set of allowed  rational output probabilties
(this also covers the case, where the output is either $0$ or $1$). In our context, the set of allowed output probabilities consists of all approximations
of some rational probability that is uniquely determined by the input stream.

Formally, a streaming problem is specified by a  function $F : \Prob^* \to 2^{\Prob}$
mapping a sequence $q_1 q_2 \cdots q_n$ of rational probabilities to a set $F(q_1 q_2 \cdots q_n)$ of rational probabilities.
Below, we slice the set of all input sequences $\Prob^*$ into the sets $\Prob_b^{\le n}$ consisting of input
sequences of length at most $n$, where in addition all input numbers have bit size at most $b$.

A deterministic streaming algorithm can be formalized by a family of deterministic finite automata $\mathcal{D} = (\mathcal{A}_{n,b})_{n,b \in \mathbb{N}}$, where the automaton $\mathcal{A}_{n,b}$ is responsible for inputs from the slice $\Prob_b^{\le n}$. Hence,
the input alphabet of $\mathcal{A}_{n,b}$ is $\Prob_b$. Let $Q_{n,b}$ be the set of states of 
$\mathcal{A}_{n,b}$. Its elements can be seen as the memory states of the streaming algorithm. There is a distinguished
initial memory state $s_{0,n,b} \in Q_{n,b}$.
The dynamics of the streaming algorithm is described by the state transition 
function $\delta_{n,b} : Q_{n,b} \times \Prob_b \to Q_{n,b}$
such that for a memory state $s \in Q_{n,b}$ and an incoming data value $q \in \Prob_b$, $\delta_{n,b}(s,q) \in Q_{n,b}$ is the new memory state.
In addition, $\mathcal{A}_{n,b}$ has an output function $o_{n,b} : Q_{n,b} \to \Prob$ instead of a set of final states.
The automaton $\mathcal{A}_{n,b}$ computes a function $f_{\mathcal{A}_{n,b}}: \Prob_b^{\le n} \to \Prob$ as follows:
Starting from the initial state, the automaton reads the input stream $w \in \Prob_b^*$. After reading
the last input number from $w$, the automaton is in a certain state $s \in Q_{n,b}$. The output value 
$f_{\mathcal{A}_{n,b}}(w)$ is then $o_{n,b}(s)$.
A deterministic streaming algorithm $\mathcal{D}= (\mathcal{A}_{n,b})_{n,b \in \mathbb{N}}$ for the streaming problem $F$ has to satisfy
$f_{\mathcal{A}_{n,b}}(w) \in F(w)$ for all $n,b \in \mathbb{N}$ and $w  \in  \Prob_b^{\le n}$.
Finally, the space complexity of the streaming algorithm $\mathcal{D}$ is the mapping $\spa_{\mathcal{D}}(n,b) = \lceil \log_2|Q_{n,b}| \rceil$,
which is the number of bits required to encode the 
states in binary notation. It is therefore the number of bits that $\mathcal{A}_{n,b}$ has to store.

Some comments regarding our definition are appropriate. First, our definition of a deterministic streaming algorithm 
is non-uniform in the sense that for every input slice $\Prob_b^{\le n}$ we have a separate algorithm (the automaton 
$\mathcal{A}_{n,b}$). In principle the function $(n, b) \mapsto \mathcal{A}_{n,b}$ may be not computable. In a real streaming
algorithm, this function will be certainly computable. On the other hand, the main focus in this paper is on lower bounds, and lower
bounds for non-uniform algorithms also hold for uniform algorithms. The few streaming algorithms that we present in this paper
are moreover clearly uniform.

Related to the previous point is the fact that we ignore in our definition of the space complexity 
$\spa_{\mathcal{D}}(n,b) = \lceil \log_2|Q_{n,b}| \rceil$ the space needed for computing the 
mapping $\delta_{n,b}$, i.e., the space
needed to compute the next memory state from the previous memory state and the input data value. This space is what we called
the space for internal computations in footnote~\ref{fn1}. In general, this internal space will also grow with $n$ and $b$. 
In a non-uniform streaming algorithm, the function $(n,b,q,s) \mapsto \delta_{n,b}(q,s)$ might be even uncomputable.
The justification for ignoring the space for internal computation is the same as for non-uniformity. 
A lower bound for our measure of space complexity also holds
for the more realistic measure that incorporates the internal space needed for computing the state transformation mapping.
Another point is that lower bound
techniques for streaming algorithms (which are typically based on information theoretic arguments or communication complexity)
can only yield lower bounds on our measure of space complexity: the classical application of communication lower bounds to streaming
lower bounds constructs from a streaming algorithm (for a certain streaming problem) a communication protocol (for a certain communication
problem). Thereby the messages exchanged between Alice to Bob are memory states of the streaming algorithm. Hence, a lower bound
for the communication problem yields a lower bound for the bit length of the memory states of the streaming algorithm. This line of argument
does not take the memory for internal computations in the streaming algorithm into account.

Randomized streaming algorithms can be defined analogously to deterministic streaming algorithms by taking
for every $\mathcal{A}_{n,b}$ a probabilistic finite automaton (see \cite{Paz71} for details on probabilistic automata). For an input stream $w \in \Prob_b^*$, we then
obtain a random variable $\mathsf{R}_{\mathcal{A}_{n,b}}^w$, whose support is the set of all output values $o_{n,b}(s)$ for $s \in Q_{n,b}$. 
For an output value $q \in \mathbb{P}$,
$\mathsf{R}_{\mathcal{A}_{n,b}}^w(q)$ is the probability that on input $w$ the probabilistic automaton 
$\mathcal{A}_{n,b}$ returns $q$. A randomized streaming algorithm $\mathcal{R} = (\mathcal{A}_{n,b})_{n,b \in \mathbb{N}}$ 
for the streaming problem $F$ has to satisfy
\begin{equation} \label{eq-rsa}
\mathsf{Prob}[\mathsf{R}_{\mathcal{A}_{n,b}}^w \in F(w)] \geq \frac{2}{3}
\end{equation}
 for all $n,b \in \mathbb{N}$ and $w  \in  \Prob_b^{\le n}$.  

We consider the following streaming problems in this paper, where $0 \leq \epsilon < 1$ is an approximation ratio
and $q_1, q_2, \ldots, q_n \in \Prob$:
\begin{itemize}
\item
The \emph{approximation of products of probabilities}, $\APP_{\epsilon}$ for short, is 
\[  
\APP_{\epsilon}(q_1 q_2\cdots q_n) = \bigg\{ q \in \mathbb{P} :   (1-\epsilon) \prod_{i=1}^n q_i < q < \frac{1}{1-\epsilon} \prod_{i=1}^n q_i \bigg\}.
\]
\item The \emph{threshold problem for products of probabilities}, $\TPP$ for short, is
\[  
\TPP(q_1q_2\cdots q_n) = \begin{cases}
1 & \text{ if } \prod_{i=2}^{n} q_i  < q_1 \\
0 & \text{ otherwise. }
\end{cases}
\]
\end{itemize}
It is easy to see that for these streaming problems the 
error probability $\leq \frac{1}{3}$ from \eqref{eq-rsa} can
be reduced to any $\lambda < 1/3$ by running $\mathcal{O}(\log(1/\lambda))$ many independent copies of 
$\mathcal{A}_{n,b}$. In the case of $\APP_{\epsilon}$, the output value will be the median of the output values
of the independent copies of $\mathcal{A}_{n,b}$, whereas for $\TPP$ (where the output is boolean) the output bit
is obtained by a majority vote.

\subsection{Communication complexity}

Lower bounds for randomized streaming algorithms are typically derived from lower bounds on randomized 
communication complexity \cite{KushilevitzN97,Roughgarden16}.  We only use the one-way setting, where Alice sends a single message to Bob.

Consider a function $f \colon X \times Y \to \{0,1\}$ for some finite sets $X$ and $Y$, which we call a 
communication problem.
A {\em randomized one-way (communication) protocol} $P$ for $f$ consists of two parties called Alice and Bob.
The input for Alice (resp., Bob) is an element $x \in X$ and a random choice $r \in R$
(resp., $y  \in Y$ and a random choice $s \in S$). Here, $R$ and $S$ are finite sets and the random choices of Alice and Bob are independent.
The goal of Alice and Bob is to compute $f(x,y)$ with high probability.
For this, Alice computes from here input $x \in X$ and random choice $r \in R$ a message $a(x,r) \in \{0,1\}^*$ and sends it
to Bob. Bob then computes from $a(x,r)$, his input $y \in Y$ and random choice $s \in S$ 
the final output $b(a(x,r),y,s) \in \{0,1\}$ of the protocol.
We also  assume a probability distribution on the set $R$ (resp., $S$)
of Alice's (resp., Bob's) random choices.
The protocol $P$ {\em computes} $f$ if 
for all $(x,y) \in X \times Y$ we have
\begin{equation} \label{eq-rcc}
\mathsf{Prob}_{r \in R, s \in S}[b(a(x,r),y,s)  = f(x,y)] \ge \frac{2}{3}.
\end{equation}
The cost of the protocol is the maximum length of $a(x,r)$ taken over all 
$(x,r) \in X \times R$.
The {\em randomized one-way communication complexity} of $f$ is the minimal cost
among all (one-way) randomized protocols that compute $f$. Here, the size of the finite sets $R$ and $S$ is not restricted.
The choice of the constant $2/3$ in \eqref{eq-rcc} is arbitrary in the sense that changing the constant
to any $1-\lambda$ only changes the communication complexity
by a constant (depending on $\lambda$),
see \cite[p.~30]{KushilevitzN97}. Also note that we only use the private version of randomized communication protocols, where Alice
and Bob make private random choices from the sets $R$ and $S$, respectively, and their choices are not known to the other
party (in contrast to the public version of randomized communication protocols).

We will use the following communication problems in this paper:
The {\em greater-than problem} $\mathsf{GT}_m$ is defined as follows:
\begin{itemize}
\item Alice's input is a number $a  \in \{1,\dots,m\}$.
\item Bob's input is a number $a' \in \{1,\dots,m\}$.
\item Bob's goal is to determine whether $a > a'$.
\end{itemize}
The {\em index-greater-than problem} $\mathsf{IGT}_{m,n}$ is defined as follows:
\begin{itemize}
\item Alice's input is a sequence $a_1 a_2 \cdots a_n$ with $a_i \in \{1,\dots,m\}$ for all $i$.
\item Bob's input is a number $a' \in \{1,\dots,m\}$ and a number $i \in \{1, \ldots, n\}$.
\item Bob's goal is to determine whether $a_i > a'$.
\end{itemize}

\begin{theorem} \label{thm:coco}
The following hold:
\begin{itemize}
\item $\mathsf{GT}_m$ has randomized one-way communication complexity $\Theta(\log m)$ {\rm \cite{MiltersenNSW98}}.
\item $\mathsf{IGT}_{m,n}$ has randomized one-way communication complexity $\Theta(n  \cdot \log m)$ {\rm \cite{BravermanGLWZ18}}.
\end{itemize}
\end{theorem}

\section{Lower bound for the product approximation problem} \label{sec lower bound approx}

In this section we consider the problem $\APP_\epsilon$.
The following upper bound was shown in \cite{AlurCJK20}.

\begin{theorem}[\mbox{\cite{AlurCJK20}}] \label{thm-alur}
For every $0 < \epsilon < 1/2$ there is a deterministic streaming algorithm for $\APP_{\epsilon}$ with
space complexity  $\log(n^2 b / \epsilon) = 2 \log n + \log b - \log \epsilon$.
\end{theorem}
Let us briefly explain the idea since it will be needed in Section~\ref{sec-SW}. Assume that the input stream is $q_1 q_2 \cdots q_n$.
Define $\epsilon' = \epsilon/n$. For an $a \in \mathbb{N}$ we define the interval $B_{a, \epsilon'} = ((1- \epsilon')^{a+1}, (1-\epsilon')^a]$, called a  \emph{bucket} in \cite{AlurCJK20}.
The streaming algorithm computes for every input number $q_i$ the unique $a_i \in \mathbb{N}$ such that $q_i \in B_{a_i, \epsilon'}$
and computes the sum $a = \sum_{i=1}^n a_i$. At the end it returns $(1-\epsilon')^a$. It can be shown that this is a $(1-\epsilon)$-approximation of the product of the $q_i$.
Moreover, every bucket index $a_i$ can be bounded by $n b/\epsilon$
(in \cite{AlurCJK20} the bound $4 n b/\epsilon$ is stated, but it 
easy to see that $n b/\epsilon$ suffices.)
 Hence, every partial sum of the $a_i$ is bounded
by $n^2 b/\epsilon$ for which $2 \log n + \log b - \log \epsilon$ bits suffice.

The space bound 
stated in \cite{AlurCJK20} also contains the space needed for computing the memory state transition function (called 
$\delta_{n,b}$ in Section~\ref{sec streaming}). This is mainly the space for computing the bucket indices $a_i$.
Recall from Section~\ref{sec streaming} that we ignore this internal space in our model.
We match the upper bound $2 \log n + \log b - \log \epsilon$ by a lower bound up to a factor of two for a reasonable range
of the parameters $n$, $b$, and $\epsilon$.

\begin{theorem} \label{thm-lb-approx-prod-det}
Let $0 < \epsilon < 1/2$. Every deterministic streaming algorithm for $\APP_{\epsilon}$ has space complexity 
at least $\log n + \log b - \log \epsilon - \Theta(1)$ when $n$, $b$, and $\epsilon$ satisfy
$-\log \epsilon \le b  <  -\log(1-\epsilon) \cdot n/2$.
\end{theorem}

\begin{proof}
Fix an $\epsilon$ with $0 < \epsilon < 1/2$ 
and let $\delta = -\log(1-\epsilon)$. Note that $0 < \delta < 1$.
Let $\mathcal{D}$ be a deterministic streaming algorithm for $\APP_{\epsilon}$.
It suffices to show that $\spa_{\mathcal{D}}(2n,b) \geq \log n + \log b - \log \delta - \Theta(1)$ 
under the assumption 
\begin{equation}
-\log \epsilon \le b  <  \delta \cdot n . \label{cond-n-eps} 
\end{equation}
We can then subsitute $n$ by $n/2$. Moreover,
it can be easily shown that the difference $\log\epsilon - \log\delta = \log\epsilon - \log\log(1/(1-\epsilon))$ 
belongs to the interval $(-1,\log\ln 2) \subseteq (-1,-0.52)$ when $0 < \epsilon < 1/2$. 
Thus, we have $\log \epsilon - \Theta(1) = \log \delta - \Theta(1)$.

By \eqref{cond-n-eps} (which yields $1/2^b \le \epsilon < 1/2$) there is a $k\in \mathbb{N}$ with $3 \le k \le 2^b$ and
\begin{equation} \label{cond-k}
\frac{1}{k} \leq \epsilon \leq  \frac{1}{k-1} 
\end{equation}
and hence
\begin{equation} \label{cond-k'}
\frac{k-2}{k-1} \leq 1-\epsilon \leq  \frac{k-1}{k} .
\end{equation}
\medskip
 \noindent
 We consider streams of length at most $2n$ of rational numbers from the set 
 \begin{equation} \label{set S}
 S = \bigg\{\frac{k-1}{k}, \frac{k-2}{k-1}, \frac{1}{2^{b}} \bigg\} \subseteq \mathbb{P}.
 \end{equation}
 Note that $\Vert q \Vert \leq b$ for all $q \in S$.
Let us define 
 \begin{equation} \label{def-y}
y := \left\lfloor \frac{n b}{\delta}\right\rfloor .
\end{equation}
For $j \ge 0$ we define the $j$-th bucket as the interval
\[
B_j = \big( (1-\epsilon)^{j+1}, (1-\epsilon)^j \big] \subseteq (0,1]. 
\]
Note that every real number $r \in (0,1]$ belongs to a unique bucket.

 \begin{claim} \label{claim-APP}
For every $j$ with $0 \leq j \leq y$ there are numbers $q_1, q_2, \ldots, q_\ell \in S$ with $\ell \le 2n$ and
$\prod_{i=1}^{\ell} q_i \in B_j$.
\end{claim}

\begin{proof}
Let $0 \leq j \leq y$. 
Choose the number $m \geq 0$ maximal such that $2^{-b m} \geq (1-\epsilon)^{j}$, which is equivalent to
$2^{bm} \leq 2^{\delta j}$, i.e., $bm \leq \delta j$.
We therefore have 
\[
m \ = \ \left\lfloor \frac{j \delta}{b} \right\rfloor \ \leq \ \frac{j \delta}{b} \leq \frac{y \delta}{b} \ \stackrel{\text{\eqref{def-y}}}{\le} \ n .
 \]
 We also have 
  \begin{eqnarray*}
(1-\epsilon)^{j} & \ \leq \ & 2^{-bm}  \ \leq \   2^{-b \left(\frac{j \cdot \delta}{b}-1\right)} 
\  = \  2^b \cdot 2^{-j \cdot \delta} 
\  =  \  2^b (1-\epsilon)^j 
\  \stackrel{\text{\eqref{cond-n-eps}}}{<} \  2^{\delta n} (1-\epsilon)^j  \\
 & \ = \ & (1-\epsilon)^{-n} (1-\epsilon)^j 
\  = \  (1-\epsilon)^{j - n}  .
 \end{eqnarray*}
It follows that if $B_s$ is the unique bucket containing $2^{-bm}$ then
$\max\{0, j-n\} \leq s \leq j$. Hence, we have  $0 \leq j-s \leq n$.
We now define the first $m \le n$ numbers of the stream as 
$q_1 = q_2 = \cdots = q_m = 2^{-b}$.
Hence we have
\begin{equation} \label{first-n-q_i}
q_1 q_2 \cdots q_m = 2^{-bm} \in B_s,
\end{equation}
where $s$ satisfies $0 \leq j-s \leq n$.

In the second step we choose  the numbers $q_{m+1}, \ldots, q_{m+j-s}$ inductively
such that $q_1 q_2 \cdots q_{m+j-s} \in B_j$. Note that $m+j-s \leq 2n$.
Assume that $q_{m+1}, \ldots, q_{m+d} \in S$ have been chosen for some $0 \leq d < j-s$ such that
$q_1 q_2 \cdots  \cdots q_{m+d} \in B_{s+d}$. Note that for $d=0$ this holds by \eqref{first-n-q_i}.
We choose $q_{m+d+1}$ by a case distinction depending on $q_1 q_2 \cdots  \cdots q_{m+d}$. It might happen that
$q_1 q_2 \cdots  \cdots q_{m+d}$ satisfies both cases 1 and 2 below. Then we can choose $q_{m+d+1}$ according to one
of them. Recall also that $k$ is such that $3 \leq k \leq 2^b$ and \eqref{cond-k'} holds. 

\medskip
\noindent
{\em Case 1.} We have 
$$
q_1 q_2 \cdots q_{m+d} \in \left( (1-\epsilon)^{s+d+1}, \frac{k}{k-1} (1-\epsilon)^{s+d+1} \right] \subseteq B_{s+d}.
$$
Inclusion in $B_{s+d}$ holds since $k (1-\epsilon)/(k-1) \leq 1$ by \eqref{cond-k'}.
Then we get 
\begin{eqnarray*}
q_1 q_2 \cdots q_{m+d} \cdot \frac{k-1}{k} & \in & \left( \frac{k-1}{k} (1-\epsilon)^{s+d+1}, (1-\epsilon)^{s+d+1} \right]  \\
& \subseteq & \left( (1-\epsilon)^{s+d+2}, (1-\epsilon)^{s+d+1} \right] = B_{s+d+1} ,
\end{eqnarray*}
where the inclusion in the second line follows from $(k-1)/k \geq  (1-\epsilon)$; see \eqref{cond-k'}. We can therefore choose $q_{m+d+1} = (k-1)/k \in S$.

\medskip
\noindent
{\em Case 2.} We have 
$$
q_1 q_2 \cdots q_{m+d} \in \left(\frac{k-1}{k-2}  (1-\epsilon)^{s+d+2}, (1-\epsilon)^{s+d} \right] \subseteq B_{s+d},
$$
where the inclusion in $B_{s+d}$ follows from $(k-1) (1-\epsilon)/(k-2) \geq 1$; see again \eqref{cond-k'}. Also 
note that $k \geq 3$ so that $k-2$ is not zero.  Then we get 
\begin{eqnarray*}
q_1 q_2 \cdots q_{m+d} \cdot \frac{k-2}{k-1} & \in & \left( (1-\epsilon)^{s+d+2}, \frac{k-2}{k-1} (1-\epsilon)^{s+d} \right] \\
& \subseteq & \left( (1-\epsilon)^{s+d+2}, (1-\epsilon)^{s+d+1} \right] = B_{s+d+1},
\end{eqnarray*}
where the inclusion in the second line follows again from \eqref{cond-k'}. We can therefore choose $q_{m+d+1} = (k-2)/(k-1) \in S$.

\medskip
\noindent
It remains to show that the cases 1 and 2 cover all possibilities, i.e., that 
\[ 
\frac{k}{k-1} (1-\epsilon)^{s+d+1} \geq \frac{k-1}{k-2}  (1-\epsilon)^{s+d+2}.
\]
After cancelling $(1-\epsilon)^{s+d+1}$ we obtain 
\begin{equation*} \label{eq-claim1}
\frac{k}{k-1}  \geq \frac{k-1}{k-2}  (1-\epsilon) = \frac{k-1}{k-2} - \epsilon \cdot \frac{k-1}{k-2} \, ,
\end{equation*}
or, equivalently,
\[
\epsilon \geq \frac{k-2}{k-1} \cdot \bigg( \frac{k-1}{k-2} - \frac{k}{k-1} \bigg) =  \frac{k-2}{k-1} \cdot  \frac{1}{(k-1) (k-2)} = \frac{1}{(k-1)^2}\ .
\]
But this is true, because by \eqref{cond-k} and $k \geq 3$ we have
$\epsilon \ge k^{-1} \ge (k-1)^{-2}$.
This concludes the proof of the claim.
\qed\end{proof}
Consider the buckets $B_{3j}$ for $0 \le j \le \lfloor y/3 \rfloor$. By our claim, for every 
$0 \le j \le \lfloor y/3 \rfloor$ there is a number $a_j \in B_{3j}$ such that $a_j$ can be written
as a product of length at most $2n$ over the set $S$. Let $s_j$ be the corresponding input 
stream whose product is $a_j$. 

Consider now $i$ and $j$ with $0 \leq i < j \le \lfloor y/3 \rfloor$.
We have $a_j < a_i$. Moreover, 
\[
\frac{a_j}{1-\epsilon} \leq \frac{(1-\epsilon)^{3j}}{1-\epsilon} = (1-\epsilon)^{3j-1} \le (1-\epsilon)^{3i+2} = (1-\epsilon)(1-\epsilon)^{3i+1} < (1-\epsilon)a_i .
\]
This means that the deterministic streaming algorithm $\mathcal{D}$ for $\APP_\epsilon$
must yield different output numbers for the input streams
$s_i$ and $s_j$. In particular, the streaming algorithm must arrive in different memory states  for the input streams
$s_i$ and $s_j$. Since this holds for all  $0 \leq i < j \le \lfloor y/3 \rfloor$, the algorithm must have at least $\lfloor y/3 \rfloor + 1 \geq y/3$ 
memory states and therefore must store at least $\log(y/3) = \log n + \log b - \log\delta - \Theta(1)$ bits.
\qed\end{proof}
We now extend Theorem~\ref{thm-lb-approx-prod-det} to randomized streaming algorithms.

\begin{theorem} 
Let $0 < \epsilon < 1/2$. 
Every randomized streaming algorithm for $\APP_{\epsilon}$ has space complexity 
$\Omega(\log n + \log b - \log \epsilon - \Theta(1))$ when $n$, $b$, and $\epsilon$ satisfy $-\log \epsilon \le b  <  -\log(1-\epsilon) \cdot n/2$.
\end{theorem}

\begin{proof}
We can show this result by a standard reduction to the one-way communication complexity of $\mathsf{GT}_m$; see Theorem~\ref{thm:coco}.
Let $\mathcal{R}$ be a randomized streaming algorithm for $\APP_{\epsilon}$. As in the proof of Theorem~\ref{thm-lb-approx-prod-det}, we define
$\delta = -\log(1-\epsilon)$. 
It suffices to show that $\spa_{\mathcal{R}}(2n,b) \geq \Omega(\log n + \log b - \log\delta -  \Theta(1))$ 
under the assumption \eqref{cond-n-eps}.
We can assume that the error probability of $\mathcal{R}$ is bounded by $1/6$. Consider the number $y$ and 
the buckets $B_{5i}$ for $0 \le i \le \lfloor y/5 \rfloor$ from the proof of Theorem~\ref{thm-lb-approx-prod-det}. Let $s_i$ be an input stream of length at most $2n$  whose product yields a number $a_i \in B_{5i}$.

We show that $\mathcal{R}$ yields a randomized one-way
protocol for $\mathsf{GT}_{\lfloor y/5 \rfloor+1}$. 
Assume that $i$ is the input for Alice and $j$ is the input for Bob, where 
$0 \le i,j \le \lfloor y/5 \rfloor$.
Alice starts by running the streaming algorithm $\mathcal{R}$ on the input stream $s_i$. Thereby she uses here random bits to simulate the probabilistic
choices of $\mathcal{R}$. Let $\alpha_i$ by the final memory state of $\mathcal{R}$ obtained by Alice. She sends $\alpha_i$ to Bob. 
Bob then runs $\mathcal{R}$ on input $s_j$; let $\alpha_j$ be the final memory state of  $\mathcal{R}$ obtained by Bob.
Bob then applies the output function of $\mathcal{R}$ to the memory states $\alpha_i$ and $\alpha_j$ and obtains rational numbers 
$r_i$ and $r_j$ respectively. With probability at least $(5/6)^2 \geq 2/3$ the output numbers $r_i$ and $r_j$ satisfy
$a_i(1-\epsilon)\leq r_i \le a_i/(1-\epsilon)$ and $a_j(1-\epsilon)\leq r_j \le a_j/(1-\epsilon)$, and thus
$r_i \in B_{5i+1} \cup B_{5i} \cup B_{5i-1}$ and $r_j \in B_{5j+1} \cup B_{5j} \cup B_{5j-1}$. Bob can then distinguish the cases
$i < j$, $i \ge j$, and $i>j$:
\begin{itemize}
\item If $i < j$ then $r_i > r_j$ and between the unique buckets containing $r_j$ and $r_i$ there are at least two buckets, namely
$B_{5i+2}, \ldots, B_{5j-2}$.
\item If $i > j$ then $r_i < r_j$ and between the unique buckets containing $r_i$ and $r_j$ there are at least two buckets, namely
$B_{5j+2}, \ldots, B_{5i-2}$.
\item If $i = j$ then $r_i, r_j \in B_{5i+1} \cup B_{5i} \cup B_{5i-1}$ and there is at most one bucket between the unique  buckets containing $r_i$ and $r_j$.
\end{itemize}
With Theorem~\ref{thm:coco},
this shows that the communicated memory state $\alpha_i$ must have bit length 
$\Omega(\log(y/5)) = \Omega(\log n + \log b - \log\delta -  \Theta(1))$.
\qed\end{proof}

\section{The threshold problem}

In this section we consider the threshold problem $\TPP$. We start with two simple upper bounds:

\begin{theorem} \label{thm-threshold-upper}
There are deterministic streaming algorithms for $\TPP$ 
with space complexity $2 n b$ and $1.443 \cdot 2^b$ (the latter for $b$ large enough), respectively.
\end{theorem}

\begin{proof}
Let $B = 2^b$. 
For notational convenience we consider streams of length $n+1$.
Consider such an input stream $q_0q_1 \cdots q_n \in \mathbb{P}^*$ of $n+1$ rational probabilities 
with $\Vert q_i \Vert \leq b$ for all $0 \le i \le n$. Recall that $q_0$ is the threshold value.
A streaming algorithm for $\TPP$ can simply store all the $q_i$ with at most
$2(n+1) b$ bits in total.

Another way to store the product $\prod_{i=1}^k q_i$ for $1 \le k \le n$
is to store for every prime $p \leq B$ how often it appears in the product, where an occurrence
in the denominator is counted negative. By the prime number theorem, the number of primes $p \leq B$
is asymptotically $B/\ln(B)$ and therefore bounded by $1.4427 \cdot B/b$ for $b$ large enough (note that $1.4427 > \log_2(e)$). The total number of occurrences of a prime $p$ in the product
is  between $-nb$ and $nb$.
This leads to the space bound $1.4427 \cdot B \cdot (\log n + \log b+1)/b$.

%

This can be improved to $1.443 \cdot B$ as follows:  
We can assume that $q_0 > 0$, otherwise the algorithm can output $0$.
Since $\Vert q_0 \Vert \leq b$ we must have $q_0 \ge 1/B$. 
Moreover, we can ignore all $1$'s in the stream $q_1 q_2 \cdots q_n$. If $q_i < 1$ then $q_i \leq (B-1)/B$.
Assume now that the number of $i$ with $1 \le i \le n$ and $q_i \leq (B-1)/B$ is at least 
$B \cdot \ln B$. Then, the product $\prod_{i=1}^n q_i$ is upper bounded by 
$$
\left(\frac{B-1}{B}\right)^{B \cdot \ln B} < \left(\frac{1}{e}\right)^{\ln B} = \frac{1}{B} \le q_0 .
$$
Hence, the algorithm can safely output $1$.
This means that we can replace in the above space bound $1.4427 \cdot B \cdot (\log n + \log b+1)/b$ the value $\log n$ by $\log(B \cdot \ln B) \leq 
b + \log(b) + \mathcal{O}(1)$ which gives the space bound $1.443 \cdot B$ for $b$ large enough
(this also includes additional space $\mathcal{O}(b)$ for storing $q_0$ and a binary counter up to $B \cdot \ln B$ for the number of $q_i \leq (B-1)/B$).
\qed\end{proof}
Let us now come to lower bounds for $\TPP$.

\begin{theorem} \label{thm-lower-bound-threshold}
Fix a constant $\gamma > 0$.
Every deterministic streaming algorithm $\mathcal{D}$ for $\TPP$ satisfies $\spa_{\mathcal{D}}(2n, (2+\gamma) (\log_2 n + b)) \geq 
n  b$ for $n$ and $b$ large enough depending on $\gamma$.
\end{theorem}

\begin{proof}
Let $B = 2^b$.
We take the first $nB+n+1$ prime numbers $p_1,p_2,...,p_{nB+n+1}$. It is known that
$p_k \leq k \cdot (\ln k + \ln \ln k)$ for $k \geq 6$ \cite[3.13]{RosserS62}.
In all numbers appearing in the two input streams that we will construct, the numerator and denominator are bounded 
by $p_{nB+n+1}^2$. Therefore, the  bit size of all the numbers involved is bounded by 
\begin{equation} \label{eq-prime bound}
\lceil \log_2 p_{nB+n+1}^2 \rceil = 
\lceil 2 \log_2 p_{nB+n+1} \rceil \leq  (2+\gamma) (\log_2 n + b)
\end{equation}
if $n$ and $b$ are large enough (depending on the constant $\gamma$).
Assume for the following that  \eqref{eq-prime bound} holds.
For every $1 \le i \le n$ we define the set
\[
\begin{matrix}
Q_i=\big \{ \frac{p_j}{p_{j+1}} :  (i-1)(B+1)+1\leq j \leq  i(B+1)-1 \big \}.
\end{matrix}
\]
Thus,  $Q_i$ consists of $B$ fractions of two  consecutive primes. Note that the sets $Q_i$ are pairwise disjoint.
We then consider the set of all words 
 \begin{equation} \label{def-S_n,b}
S_{n,b} = \{ q_1 q_2 \cdots q_n \in \mathbb{P}^n : q_i \in Q_i \text{ for all } 1 \leq i \leq n\} . 
\end{equation}
Clearly, $|S_{n,b}| = B^n$. 
Moreover, for two different words $q_1 q_2 \cdots q_n, r_1 r_2 \cdots r_n \in S_{n,b}$ the products 
$\prod_{i=1}^n q_i$ and $\prod_{i=1}^n r_i$ are also different due to the uniqueness of prime factorizations.
We set the threshold (the first number in the stream) to 
\begin{equation} \label{def-t}
 t =\frac{p_1}{p_{nB+n+1}}.
\end{equation}
Assume now that there is a deterministic streaming algorithm $\mathcal{D}$ for  $\TPP$ such that
$\spa_{\mathcal{D}}(2n, (2+\gamma) (\log_2 n + b)) < n \cdot b = \log_2 |S_{n,b}|$. 
Then there must exist two different words $u = q_1 q_2 \cdots q_n$ and
$v = r_1 r_2 \cdots r_n$ in $S_{n,b}$ such that 
after reading $tu$ and $tv$, $\mathcal{D}$ is in the same memory state. 
W.l.o.g.~we assume  that $\prod_{i=1}^n q_i < \prod_{i=1}^n r_i$.
We continue the words $t u$ and $t v$, respectively, with the length-$n$ word
$w = (s_1/r_1) (s_2/r_2) \cdots (s_n/r_n)$, where 
\begin{equation} \label{def-s_i}
s_i=\frac{p_{(i-1)(B+1)+1}}{p_{i(B+1)+1}} .
\end{equation}
Note that $s_i/r_i \leq 1$ and that
\[
\prod_{i=1}^n  s_{i} = \prod_{i=1}^n \frac{p_{(i-1)(B+1)+1}}{p_{i(B+1)+1}}  = \frac{p_1}{p_{nB+n+1}} = t .
\]
 Finally, we have
\begin{eqnarray}
\prod_{i=1}^n q_i \prod_{i=1}^n  r_i^{-1}\prod_{i=1}^n  s_{i} & < & \prod_{i=1}^n  s_i = t \text{ and } \\
\prod_{i=1}^n r_i \prod_{i=1}^n  r_i^{-1}\prod_i^n s_i & = & \prod_{i=1}^n  s_{i}  = t .
\end{eqnarray}
But this is a contradiction, since $\mathcal{D}$ arrives in the same memory state after reading the streams $tuw$ and $tvw$.
\qed\end{proof}
From the construction in the proof of Theorem~\ref{thm-lower-bound-threshold} we can easily obtain a lower
bound for randomized streaming algorithms for $\TPP$:

\begin{theorem} \label{thm random TPP 1}
Fix a constant $\gamma > 0$.
Every randomized streaming algorithm $\mathcal{R}$ for $\TPP$ satisfies $\spa_{\mathcal{R}}(2n, (2+\gamma) (\log_2 n + b)) \geq 
\Omega(n \cdot b)$.
\end{theorem}

\begin{proof}
We transform a randomized streaming algorithm $\mathcal{R}$ for $\TPP$ 
into a randomized one-way protocol for $\mathsf{GT}_{B^n}$. 

Take the set of words $S_{n,b}$ from  \eqref{def-S_n,b}, the numbers $s_i$ from  \eqref{def-s_i}
and the threshold $t$ from \eqref{def-t}.
Let $f : \{1,2,\ldots,B^n\} \to S_{n,b}$ be the unique bijection such that for all $1 \le i,j \le B^n$ with $f(i) = q_1 q_2 \cdots q_n$
and $f(j) = r_1 r_2 \cdots r_n$ we have:
$i > j$ if and only if $\prod_{k=1}^n q_k < \prod_{k=1}^n r_k$.

On input $i$, Alice computes the word $f(i) = q_1 q_2 \cdots q_n$, reads $t \, q_1 \cdots q_n$ into $\mathcal{R}$ and sends 
the resulting memory state $\alpha$ of $\mathcal{R}$ to Bob. Bob then computes from his input $j$ the word $f(j) = r_1 r_2 \cdots r_n$
and reads, starting from memory state $\alpha$, the word $(s_1/r_1)(s_2/r_2) \cdots (s_n/r_n)$ into the algorithm.
From the final output of $\mathcal{R}$, 
Bob can decide with probability at least $2/3$ whether $i > j$:
\begin{itemize}
\item If $\prod_{k=1}^n q_k <  \prod_{k=1}^n r_k$ (i.e., $i > j$) then $\mathcal{R}$  outputs $1$ with probability at least $2/3$ because
\[
\prod_{k=1}^n q_k \prod_{k=1}^n  \frac{s_k}{r_k} <  \prod_{k=1}^n  s_k = t .
\]
\item If $\prod_{k=1}^n q_k \ge  \prod_{k=1}^n r_k$ (i.e., $i \le j$) then $\mathcal{R}$  outputs $0$ with probability at least $2/3$ because
\[ \prod_{k=1}^n q_k \prod_{k=1}^n  \frac{s_k}{r_k}  \ge  \prod_{k=1}^n  s_k = t .\]
\end{itemize}
It follows from Theorem~\ref{thm:coco} that the communicated memory state $\alpha$ of $\mathcal{R}$ must have bit length $\Omega(\log B^n) = \Omega(n \cdot b)$.
\qed\end{proof}
If we require $b \ge \Omega(\log_2 n)$ then the term $(2+\gamma) (\log_2 n + b)$ in Theorem~\ref{thm random TPP 1} becomes
$\Theta(b)$ and we obtain:

\begin{corollary} \label{thm random TPP 2}
Every randomized streaming algorithm $\mathcal{R}$ for the threshold problem $\TPP$ satisfies 
$\spa_{\mathcal{R}}(n,b) \geq \Omega(n \cdot b)$ whenever $b \ge \Omega(\log_2 n)$.
\end{corollary}
The condition $b \ge \Omega(\log_2 n)$ cannot be completely avoided in Corollary~\ref{thm random TPP 2}, since 
this would contradict the upper bound $1.443 \cdot 2^b$ in Theorem~\ref{thm-threshold-upper}.

Also note that the lower bound for $\APP_\epsilon$ in Theorem~\ref{thm-lb-approx-prod-det} holds for input streams over a set of only 3 numbers
(namely the set $S$ in \eqref{set S}). The lower bound from Corollary~\ref{thm random TPP 2} does not hold for a constant number of probabilities, since we can then store the product with $\mathcal{O}(\log n)$ 
bits.

\section{Sliding window products} \label{sec-SW}

In this section we consider the following sliding window variant 
 $\SWAPP_{m,\epsilon}$ of $\APP_\epsilon$, where $m$ is the size of the sliding window.
\[  
\SWAPP_{m,\epsilon}(q_1 q_2\cdots q_n) = \bigg\{ q \in \mathbb{P} :   (1-\epsilon) \prod_{i=n-m+1}^n q_i < q < \frac{1}{1-\epsilon} \prod_{i=n-m+1}^n q_i \bigg\}
\]
(here, we set $q_i = 1$ for $i \le 0$).  A randomized sliding window algorithm for $\SWAPP_{m,\epsilon}$ is a collection of probabilistic finite automata
$\mathcal{R} = (\mathcal{A}_b)_{b \in \mathbb{N}}$ such that for every $b \in \mathbb{N}$ and every input stream $w \in \mathbb{P}_b^*$ we have:
\[
\mathsf{Prob}[\mathsf{R}_{\mathcal{A}_{b}}^w \in \SWAPP_{m,\epsilon}(w)] \geq \frac{2}{3}.
\]
As usual for sliding window algorithms, we measure the space complexity of a sliding window algorithm for $\SWAPP_{m,\epsilon}$
 in the window size $m$ (instead of
the total stream length $n$), the bit size $b$ and the approximation ratio $\epsilon$.

\begin{theorem} \label{thm SW upper bound}
For $0 < \epsilon < 1/2$, there are
deterministic sliding window algorithms for $\SWAPP_{m,\epsilon}$ with
space complexity $m \cdot (\log m + \log b - \log \epsilon)$ and $2 m b$, respectively.
\end{theorem}

\begin{proof}
The upper bound $2 m b$ is clear, since the whole window content can be stored with $2mb$ bits.
 For the upper bound $m \cdot (\log m + \log b - \log \epsilon)$
we use the algorithm for $\APP_\epsilon$ from \cite{AlurCJK20}. Recall the sketch after Theorem~\ref{thm-alur}.
Our sliding window algorithm stores for a window content $q_1 q_2 \cdots q_m$ the sequence of bucket indices
$a_1a_2 \cdots a_m$, where $q_i \in B_{a_i, \epsilon'}$ for $\epsilon' = \epsilon/m$. By the analysis in \cite{AlurCJK20}, $(1-\epsilon')^a$ (with
$a = \sum_{i=1}^m a_i)$ is a $(1-\epsilon)$-approximation of $\prod_{i=1}^m q_i$. Moreover, every $a_i$ is bounded by 
$mb/\epsilon$ and therefore can be stored with $\log m + \log b - \log \epsilon$ bits.
\qed\end{proof}

\begin{theorem} \label{thm SW lower bound}
Let $0 < \epsilon < 1/2$. 
Every randomized sliding window algorithm for $\SWAPP_{m,\epsilon}$
requires space $\Omega(m \cdot (\log b - \log \epsilon - \Theta(1)))$.
\end{theorem}

\begin{proof}
Let $\delta := -\log(1-\epsilon) \in (0,1)$. Since $\log \epsilon - \Theta(1) = \log \delta - \Theta(1)$ (see the beginning 
of the proof of Theorem~\ref{thm-lb-approx-prod-det}),
it suffices to show the space lower bound 
$\Omega(m \cdot (\log b - \log \delta - \Theta(1)))$.

Define the following integers:
\begin{equation}
\alpha = \lceil 4 \delta \rceil, \qquad c = \left\lfloor\frac{b}{\alpha}\right\rfloor .
\end{equation}
Consider the set $A = \{ 2^{-i \cdot \alpha} : 0 \le i \le c-1 \} \subseteq \mathbb{P}$ of size $c$.  
Also note that $\Vert q \Vert \leq c \cdot \alpha \le b$ for all $q \in A$. We prove the theorem by a reduction from the communication
problem $\mathsf{IGT}_{c,m}$. 

Let $\mathcal{R}$ be a randomized sliding window algorithm for $\SWAPP_{m,\epsilon}$.
We can assume that the error probability of $\mathcal{R}$ is bounded by $1/6$.
We obtain a randomized one-way communication protocol for $\mathsf{IGT}_{c,m}$ as follows: 
The input of Alice is a sequence $a_1, a_2, \ldots, a_m$ with $a_i \in A$ and the input of Bob is an index $1 \le i \le m$ and a number
$a \in A$. Bob wants to find out whether $a_i > a$. Alice sends the memory state of the randomized sliding window algorithm after reading
$a_1a_2\cdots a_m$ to Bob. This allows Bob to compute with high probability $(1-\epsilon)$-approximations $P'_1$ and $P'_2$ of the products
\begin{equation*}
P_1 = a_i a_{i+1} \cdots a_m \quad \text{ and } \quad P_2 = a_{i+1} \cdots a_m a,
\end{equation*} 
respectively.
For this, Bob continues the sliding window algorithm from the memory state obtained from Alice with the input numbers $1,1, \ldots, 1$ 
and $1,1, \ldots, 1, a$, respectively, where the number of $1$'s is $i-1$. 
With probability at least $1 - 1/3$ the computed approximations $P'_1$ and $P'_2$ satisfy 
\begin{equation} \label{eq P'_1, P'_2}
(1-\epsilon) P_1 < P'_1 < P_1/(1-\epsilon)   \quad \text{ and } \quad  (1-\epsilon) P_2 < P'_2 < P_2/(1-\epsilon) .
\end{equation}
Assume for the following that \eqref{eq P'_1, P'_2} holds.
We claim that from the quotient $P'_1/P'_2$, Bob can determine whether $a_i > a$ holds. Note that $P_1/P_2 = a_i/a$.
Moreover, \eqref{eq P'_1, P'_2} implies
\begin{equation*}
 (1-\epsilon)^2 \cdot \frac{a_i}{a} \ < \ \frac{P'_1}{P'_2} \ < \ \frac{1}{(1-\epsilon)^2} \cdot  \frac{a_i}{a} .
\end{equation*}
If $a_i \le a$ then $P'_1/P'_2 < 1/(1-\epsilon)^2$. On the other hand, if $a_i > a$ then by the choice of the set $A$ we must have
$a_i/a \ge 2^{\alpha} \ge 2^{4 \delta} = 1/(1-\epsilon)^4$, which yields $P'_1/P'_2 > 1/(1-\epsilon)^2$. Hence, from the quotient $P'_1/P'_2$
Bob can indeed distinguish the cases $a_i > a$ and $a_i \le a$.
By Theorem~\ref{thm:coco}, $\mathcal{R}$ must use
space $\Omega(m \cdot \log c) = \Omega(m \cdot (\log b - \log \delta - \Theta(1)))$.
\qed\end{proof}

\section{Open problems}

For the space complexity of $\SWAPP_{m,\epsilon}$, there is a gap 
between the upper bounds in Theorem~\ref{thm SW upper bound} and the lower bound from Theorem~\ref{thm SW lower bound}
that we would like to close. 
For $\APP_{\epsilon}$, a small gap (by a multiplicative factor of 2) arises from the upper bound 
in Theorem~\ref{thm-alur} and the lower bound in Theorem~\ref{thm-lb-approx-prod-det}.
One may also ask, whether the condition $-\log \epsilon \le b  <  -\log(1-\epsilon) \cdot n/2$ in Theorem~\ref{thm-lb-approx-prod-det} can be relaxed.


\begin{thebibliography}{10}

\bibitem{Aggarwal07}
Charu~C. Aggarwal, editor.
\newblock {\em Data Streams - Models and Algorithms}, volume~31 of {\em
  Advances in Database Systems}.
\newblock Springer, 2007.
\newblock \href {https://doi.org/10.1007/978-0-387-47534-9}
  {\path{doi:10.1007/978-0-387-47534-9}}.

\bibitem{AlurCJK20}
Rajeev Alur, Yu~Chen, Kishor Jothimurugan, and Sanjeev Khanna.
\newblock Space-efficient query evaluation over probabilistic event streams.
\newblock In {\em Proceedings of the 35th Annual {ACM/IEEE} Symposium on Logic
  in Computer Science, {LICS} 2020}, pages 74--87. {ACM}, 2020.
\newblock \href {https://doi.org/10.1145/3373718.3394747}
  {\path{doi:10.1145/3373718.3394747}}.

\bibitem{BHK2020}
Avrim Blum, John Hopcroft, and Ravindran Kannan.
\newblock {\em Foundations of Data Science}.
\newblock Cambridge University Press, 2020.
\newblock \href {https://doi.org/10.1017/9781108755528}
  {\path{doi:10.1017/9781108755528}}.

\bibitem{BravermanGLWZ18}
Vladimir Braverman, Elena Grigorescu, Harry Lang, David~P. Woodruff, and Samson
  Zhou.
\newblock Nearly optimal distinct elements and heavy hitters on sliding
  windows.
\newblock In {\em Proceedings of the 21st International Conference on
  Approximation Algorithms for Combinatorial Optimization Problems, and the
  22nd International Conference on Randomization and Computation,
  {APPROX/RANDOM} 2018}, volume 116 of {\em LIPIcs}, pages 7:1--7:22. Schloss
  Dagstuhl - Leibniz-Zentrum f\"ur Informatik, 2018.
\newblock \href {https://doi.org/10.4230/LIPIcs.APPROX-RANDOM.2018.7}
  {\path{doi:10.4230/LIPIcs.APPROX-RANDOM.2018.7}}.

\bibitem{DatarM02}
Mayur Datar and S.~Muthukrishnan.
\newblock Estimating rarity and similarity over data stream windows.
\newblock In {\em Proceedings of the 10th European Symposium on Algorithms,
  {ESA} 2002}, volume 2461 of {\em Lecture Notes in Computer Science}, pages
  323--334. Springer, 2002.
\newblock \href {https://doi.org/10.1007/3-540-45749-6_31}
  {\path{doi:10.1007/3-540-45749-6_31}}.

\bibitem{KushilevitzN97}
Eyal Kushilevitz and Noam Nisan.
\newblock {\em Communication complexity}.
\newblock Cambridge University Press, 1997.
\newblock \href {https://doi.org/10.1017/CBO9780511574948}
  {\path{doi:10.1017/CBO9780511574948}}.

\bibitem{MiltersenNSW98}
Peter~Bro Miltersen, Noam Nisan, Shmuel Safra, and Avi Wigderson.
\newblock On data structures and asymmetric communication complexity.
\newblock {\em Journal of Computer and System Sciences}, 57(1):37--49, 1998.
\newblock \href {https://doi.org/10.1006/JCSS.1998.1577}
  {\path{doi:10.1006/JCSS.1998.1577}}.

\bibitem{Muthukrishnan05}
S.~Muthukrishnan.
\newblock Data streams: Algorithms and applications.
\newblock {\em Foundations and Trends in Theoretical Computer Science}, 1(2),
  2005.
\newblock \href {https://doi.org/10.1561/0400000002}
  {\path{doi:10.1561/0400000002}}.

\bibitem{Paz71}
Azaria Paz.
\newblock {\em Introduction to Probabilistic Automata (Computer Science and
  Applied Mathematics)}.
\newblock Academic Press, Inc., 1971.
\newblock \href {https://doi.org/10.1016/C2013-0-11297-4}
  {\path{doi:10.1016/C2013-0-11297-4}}.

\bibitem{RosserS62}
J.~Barkley Rosser and Lowell Schoenfeld.
\newblock Approximate formulas for some functions of prime numbers.
\newblock {\em Illinois Journal of Mathematics}, 6(1):64--94, 1962.
\newblock \href {https://doi.org/10.1215/ijm/1255631807}
  {\path{doi:10.1215/ijm/1255631807}}.

\bibitem{Roughgarden16}
Tim Roughgarden.
\newblock Communication complexity (for algorithm designers).
\newblock {\em Foundations and Trends in Theoretical Computer Science},
  11(3-4):217--404, 2016.
\newblock \href {https://doi.org/10.1561/0400000076}
  {\path{doi:10.1561/0400000076}}.

\end{thebibliography}

 \end{document}